\documentclass[11pt,numbers]{article}
\pdfoutput=1
\usepackage[margin=1in]{geometry}
\usepackage{authblk}
\usepackage{fontenc}
\usepackage{amsmath}
\usepackage{amsfonts,amssymb,amsthm, bbm}
\usepackage{hyperref}
\usepackage[capitalize]{cleveref}
\usepackage{mathtools}
\usepackage{xcolor}
\usepackage{tikz}

\hypersetup{
    colorlinks,
    linkcolor={red!50!black},
    citecolor={blue!50!black},
    urlcolor={blue!80!black}
}
\usepackage{graphicx}   
\usepackage{subfigure}  
\usepackage{amsbsy} 
\usepackage[bold]{hhtensor}
\usepackage{natbib}
\usepackage{times}
\usepackage{multirow}
\newtheorem{theorem}{Theorem}

\newtheorem{lemma}[theorem]{Lemma}
\newtheorem{proposition}[theorem]{Proposition}
\newtheorem{definition}[theorem]{Definition}

\newcommand{\myleft}{\mathopen{}\mathclose\bgroup\left}
\newcommand{\myright}{\aftergroup\egroup\right}

\renewcommand{\i}{\ensuremath\mathrm{i}}

\DeclareMathOperator{\Tr}{Tr}

\newcommand{\ket}[1]{\left.\left|{#1}\right.\right\rangle}

\newcommand{\bra}[1]{\left.\left\langle{#1}\right.\right|}

\newcommand{\ketbra}[2]{\ket{#1} \!\! \bra{#2}}

\newcounter{example}[section]

\DeclarePairedDelimiter{\ceil}{\lceil}{\rceil}

\allowdisplaybreaks[4]

\newcommand{\bits}{\{0,1\}}

\renewcommand{\ketbra}[1]{\ket{#1}\bra{#1}}
\newcommand{\fm}[1]{#1^{\mathrm{(f)}}}
\newcommand{\qb}[1]{#1^{\mathrm{(q)}}}

\newcommand{\qf}[1]{#1^{\mathrm{(qf)}}}

\newcommand{\bs}{\backslash}

\begin{document}
\title{Fermionic Hamiltonians without trivial low-energy states}

\author[,1,2,3]{Yaroslav Herasymenko\thanks{Corresponding author: yaroslav@cwi.nl}}

\author[4]{Anurag Anshu}

\author[1,2]{Barbara~M.~Terhal}

\author[3]{Jonas Helsen}

\affil[1]{QuTech, TU Delft, P.O. Box 5046, 2600 GA Delft, The Netherlands}
\affil[2]{Delft Institute of Applied Mathematics, TU Delft, 2628 CD Delft, The Netherlands}
\affil[3]{QuSoft and CWI, Science Park 123, 1098 XG Amsterdam, The Netherlands}
\affil[4]{School of Engineering and Applied Sciences, Harvard University, 150 Western Ave., Allston, MA 02134, US}
\maketitle
\begin{abstract}
We construct local fermionic Hamiltonians with no low-energy trivial states (NLTS), providing a fermionic counterpart to the NLTS theorem. Distinctly from the qubit case, we define trivial states via finite-depth \textit{fermionic} quantum circuits. We furthermore allow free access to Gaussian fermionic operations, provided they involve at most $O(n)$ ancillary fermions. The desired fermionic Hamiltonian can be constructed using any qubit Hamiltonian which itself has the NLTS property via well-spread distributions over bitstrings, such as the construction in~\cite{anshu2022nlts}. 
We define a fermionic analogue of the class quantum PCP and discuss its relation with the qubit version.
\end{abstract}

\section{Introduction}

One of the main problems in quantum complexity is establishing the hardness of quantum computational problems such as the estimation of the ground state energy of a many-body Hamiltonian. If high accuracy is required, --scaling inversely polynomial with the problem size $n$--, this problem is known to be QMA-complete \cite{KSV}. It has been an open question 
whether the estimation of the ground state energy {\em density} of a Hamiltonian with constant error, would also be QMA-hard \cite{AAV}. Such a result would mirror the classical PCP theorem, which states that there are classical SAT problems for which it is NP-complete to decide whether all clauses can be satisfied or a constant fraction of the clauses can never be satisfied.\\

Motivated by this quantum PCP problem, it was conjectured in \cite{freedman2013quantum} that there exist local sparse qubit Hamiltonians such that any state with sufficiently small energy density cannot be prepared by a constant-depth circuit, hence suggesting that low-energy states have some nontrivial complexity. This ``no low-energy trivial state'' (NLTS) conjecture was recently proven in \cite{anshu2022nlts}. They demonstrate that a family of Hamiltonians based on recently established good quantum LDPC codes \cite{leverrier2022quantum} has the desired property. While the low energy states of this family of Hamiltonians are non-trivial, their ground state can be efficiently prepared by a Clifford circuit. This means that these Hamiltonians fall short of being genuinely `complex'. Hence it is natural to extend the definition of ``non-trivial" beyond constant-depth circuits to climb up the complexity ladder, see \cite{GL22,WFC:guide, coble2023local,anschuetz2023combinatorial} . In particular, ~\cite{coble2023local} have shown that an adaptation of the construction in \cite{anshu2022nlts} can be used to construct Hamiltonians without either stabilizer low-energy states or constant-depth low-energy states.\\

Here we similarly extend the construction in \cite{anshu2022nlts}, but now in the direction of fermionic Hamiltonians.
Fermionic Hamiltonians describe the majority of quantum systems of practical interest, namely interacting electrons in materials and chemical compounds. These systems are among the prime targets for quantum simulation using both classical and quantum computers. Fermionic quantum complexity is interesting when it differs from qubit complexity and when qubit arguments are not directly applicable. In addition, while most attention in the quantum complexity literature is devoted to qubit systems, fermions have come more into focus in recent years, see e.g.~\cite{BG:impurity,BGKT:manybody,oszmaniec+:FS,hastings21syk,HSHT}. \\

 In the present manuscript, we answer the question posed in \cite{hastings21syk} in the affirmative: do there exist fermionic Hamiltonians without low-energy trivial states? Our results directly build on \cite{anshu2022nlts}. We give natural notions of triviality through the class of fermionic Gaussian states, which are classically simulatable
 Besides our main result, the fermionic NLTS Theorem (Theorem \ref{thm:fNLTS}), we analyze the definition of the quantum PCP class from a fermionic perspective, presenting a fermionic version (Section \ref{sec:PCP}). The paper is otherwise organized as follows. In Section \ref{sec:circuits} we provide some useful definitions, in particular we give two different natural notions of trivial fermionic circuits.  In Section \ref{sec:well-spread} we introduce the key tools from \cite{anshu2022nlts} and in Section \ref{sec:fermionic_nlts} we prove our main Theorem. We end the paper with a Discussion (Section \ref{sec:discuss}) on sparse fermion-to-qubit mappings and what lies beyond.

\section{Preliminaries and Definitions}
\label{sec:circuits}

A system of $n$ fermionic modes can be described by $2n$ Majorana fermion operators $c_i$ ($i\in [2n]$) which obey $c_i^2=\mathbb{I}$, $\forall \i \neq j\; c_i c_j=-c_i c_j$, $c_i=c_i^{\dagger}$, ${\rm Tr}(c_i)=0$.
Any fermionic state on this system is given by a Hermitian positive semidefinite $\rho\geq 0$ with ${\rm Tr}(\rho)=1$ which can be expressed as an {\em even} polynomial in the Majorana fermion operators.

We define Hermitian operators $C_K$ as ordered products of the operators $c_k$:
    \begin{align}
        C_{K}=i^{|K|(|K|-1)/2} c_{k_1} c_{k_2} ..c_{k_{|K|}},
    \end{align}
with $K=(k_1< k_2 < \ldots < k_{|K|})$. 
The class of Gaussian fermionic states is a subclass of fermionic states which are efficiently describable. Gaussian circuits preserve this class of states, and hence are efficiently-classically simulatable.  For more background on Gaussian states and Gaussian circuits, --also called fermionic linear optics--, we refer the reader to e.g. \cite{bravyi00ferm,Terhal_2002,bravyi2004lagrangian,hastings21syk,HSHT}. 

Consider preparing a fermionic state $\rho$ on $n$ modes using $m$ ancillary modes (with Majorana operators $c_k$, $k\in[2(n+m)]\bs[2n]$). Starting from a \textit{pure} initial Gaussian state $\sigma_G$, one may use a fermionic circuit $U$ which is generally decomposable into gates which use quartic or quadratic interactions between O(1) fermionic operators \cite{bravyi00ferm}, i.e.
\begin{equation}
\rho=\Tr_{\left[2(n+m)\right]\bs\left[2n\right]}\left(U
\sigma_G U^\dag\right).
\label{eq:fermi_state_prep}
\end{equation}
We now introduce two slightly different notions of nontrivial fermionic circuit `depth'. Our definition will make use of the notion of Gaussian circuits, namely circuits consisting of unitaries of the form $\exp\left(- \omega \,c_{k_1} c_{k_2}\right)$. We note that any Gaussian unitary on $n$ fermionic modes can be written as a $\Theta(n)$-deep circuit \cite{Jiang_2018} of such unitaries. The definition is as follows:

\begin{definition}
    \label{def:fermi_complexity}
    We say that a fermionic circuit $U$ has depth $T$, if it is given by
    \begin{equation}
        U=W_{T} W_{T-1}..W_1,
        \label{eq:fermion_circuit}
    \end{equation}
    with each $W_t$ of the form
    \begin{align}
    W_t=\prod_{K\colon|K|\in\{2,4\}} \exp\left(i \omega^{(t)}_{K}\,C_K\right)
    \label{eq:nongaussian_operation}
    \end{align}
    where the product runs over \emph{non-overlapping} sets $K\subset [2n]$. Furthermore we say that $U$ is a depth $T$ circuit with \emph{free access to Gaussian operations}, if it is given by
    \begin{align}
        U=G_{T}W_TG_{T-1}..W_2G_1W_1,
        \label{eq:fermion_circuit_free_gauss}
    \end{align}
    where $G_t$ are arbitrary Gaussian circuits.
\end{definition}

Note we take a fixed choice of operators $c_s$ in Eq.\,\eqref{eq:nongaussian_operation}. Allowing the use of an arbitrary basis of fermionic operators $\{c_k\}$ in each layer of $W_t$ would have allowed more expressive unitaries $U$. However, we don't consider this model separately, since it is subsumed by free access to Gaussian operations in Eq.\,\eqref{eq:fermion_circuit_free_gauss}.

Here is another useful (standard) definition:

\begin{definition}[Local sparse $n$-qubit (resp. $n$-fermion) Hamiltonian]
    $H=\sum_{i=1}^m H_i$ is a local sparse $n$-qubit (resp. $n$-fermion) Hamiltonian when the maximum number of Pauli operators (resp. Majorana operators) in each term $H_i$ is $O(1)$ (local) and the maximum number of terms involving any qubit (resp. Majorana operator) is $O(1)$. The smallest eigenvalue of $H$ is denoted as $\lambda(H)$.
\end{definition}

Throughout this manuscript, we will be considering Hamiltonians which are sparse, local, positive-semidefinite and have terms bounded by one in operator norm unless explicitly said otherwise.

\section{NLTS Hamiltonians and well-spread distributions}
\label{sec:well-spread}

It was proven in \cite{anshu2022nlts} that there exist qubit Hamiltonians with the following NLTS property
\begin{definition}
    A local sparse $n$-qubit Hamiltonian $H\geq 0$ has the NLTS property with parameter $\epsilon$, if any family of $n$-qubit states $\rho_n$ with energy $\mathrm{Tr} ({\rho}_n {H})<\epsilon n$ requires a quantum circuit, which uses an arbitrary number of ancilla qubits, of depth at least $T=\Omega(\log n)$. 
\end{definition}

A way to obtain NLTS Hamiltonians passes through the notion of well-spread quantum states, which are states that, when measured in a ``trivial" basis, have certain statistical properties. We now define well-spreadness both for qubits and fermions; the fermionic definition will play a key role in our result.

\begin{definition}
\label{def:loc_measurement}
   Consider an $n$-qubit POVM $\pi_{R}(s)$, defined as 
    \begin{equation}
    \pi_{R}(s)=R^\dag\prod^n_{j=1}\frac{1+(-1)^{s_j}Z_j}{2}R,
    \end{equation}
    where $s$ is the outcome bitstring $s=(s_1,..s_n)\in \bits^n$ and $R$ is a tensor product of single qubit unitaries. 
    Given an $n$-qubit state $\rho$, define a probability distribution over $s\in \bits^n$,
    \begin{equation}
    p_{R,\rho}(s)=\Tr(\pi_{R}(s)\rho).
    \end{equation}
\end{definition}

\begin{definition}
    \label{def:gaus_measurement}
    Consider a POVM with $2n$ Majorana fermion operators 
    \begin{equation}
    \pi_{G}(s)=G^\dag \prod^n_{j=1}\frac{\mathbb{I}+(-1)^{s_j}c_{2j-1}c_{2j}}{2}G,
    \end{equation}
    where $G$ is a Gaussian unitary and $s$ is the outcome bitstring. A state $\rho$ on $n$ fermionic modes yields a probability distribution over $s\in \bits^n$
    \begin{equation}
    p_{G,\rho}(s)=\Tr(\pi_{G}(s)\rho).
    \end{equation}

\end{definition}
    
\begin{definition}
    \label{def:well-spread}
    A qubit (resp. fermionic) state $\rho$ is said to be $(\mu, L)$\emph{-spread} if exists a unitary $R$ (respectively, $G$) as defined above, and two sets of bitstrings $S_1$ and $S_2\subset \{0,1\}^n$ such that $p(s)\equiv p_{R,\rho}(s)$ (respectively, $p_{G,\rho}(s)$) obeys
    \begin{equation}
    \sum_{s\in S_1} p(s)\geq \mu,~~\sum_{s\in S_2} p(s)\geq \mu,~~\mathrm{and}~~\min_{s_{1,2}\in S_{1,2}}|s_1-s_2|=L.
    \end{equation}
    with Hamming distance $|.|$. 
    An $n$-qubit ($n$-fermion) Hamiltonian $H$ is referred to as having well-spread low-energy states, if there exists a constant $\epsilon$ such that any $\rho$ with energy $\mathrm{tr}(\rho H)<\epsilon n$ is $(\mu, L)$\emph{-spread} for $\mu=\Omega(1)$ and $L=\Omega(n)$.
\end{definition}

Ref.\,\citep{anshu2022nlts} gives a quantum LDPC (CSS) code construction of an $n$-qubit Hamiltonian $H$ with well-spread low-energy states. This property implies that $H$ is NLTS (Fact 4 in \cite{anshu2022nlts}).

\section{Fermionic NLTS Hamiltonians}
\label{sec:fermionic_nlts}

We show that qubit Hamiltonians with well-spread low energy states (by Definition \ref{def:well-spread}) can be used to construct fermionic local Hamiltonians without trivial low-energy states. 
Terminology related to the fermionic circuits was introduced in Definition~\ref{def:fermi_complexity}.

\begin{theorem}
\label{thm:fNLTS}
For even $n$, consider an $n$-qubit $\qb{H}$ with well-spread states below energy $\epsilon n$. Construct a $3n/2$-fermion Hamiltonian $H$ with Majorana operators $\{c_{x,j},~c_{y,j},~c_{z,j}\}_{j\in[n]}$, by replacing the Pauli operators in $\qb{H}$ using
\begin{align}
        X_j\mapsto i c_{y,j} c_{z,j},~~
        Y_j\mapsto i c_{x,j} c_{z,j},~~
        Z_j& \mapsto i c_{x,j} c_{y,j}.
        \label{eq:qam_thm}
\end{align}
The resulting Hamiltonian $H$ has two fermionic NLTS properties.
In particular, for an arbitrary $3n/2$-fermion state ${\rho}$ such that $\mathrm{Tr} ({\rho} {H})<\epsilon n$, the following holds
\begin{enumerate}
    \item 
    Using arbitrary Gaussian initialization and an arbitrary number of ancillas, any fermionic circuit that prepares $\rho$ has depth $T = \Omega(\log(n))$
    \item Using a fermionic circuit with free access to Gaussian operations and at most $m=O(n)$ ancillary fermionic modes, any circuit that prepares $\rho$ has depth $T=\Omega(\log n)$.
\end{enumerate}
\end{theorem}

Before proceeding to the proof, we would like to briefly discuss the above result. The restriction on the access to ancillary modes given access to Gaussian operations, which we see in Theorem \ref{thm:fNLTS}.2, may be of fundamental origin. On the one hand, ancillary non-Gaussian states in combination with Gaussian operations allow the injection of non-Gaussian gates \cite{bravyi:uni}. On the other hand, this procedure also requires a capacity for adaptive measurements, which is not included in our scenario. We leave as an open question whether the particular upper bound $m=O(n)$ is optimal, or if it can be improved using other proof techniques.\\

Our constructed Hamiltonian $H$ has the same spectrum as $\qb{H}$, with an additional $2^{n/2}$-fold degeneracy. As will be clear from the proof, the degeneracy does not affect the well-spreadness property. The key point is that the mapping we use preserves the locality of $H$ (unlike, say, a direct mapping of $n$ qubits onto $n$ fermions via the inverse Jordan-Wigner transformation). We believe that a similar result can be achieved with the local embedding of $n$ qubits into $4n$ Majorana operators, as introduced by \cite{kitaev06}, or a mapping used in \cite{Liu_2007}.

\begin{proof}  

    The Hamiltonian $H$ constructed using Eqs.\,\ref{eq:qam_thm} can be understood as an image of $\qb{H}\otimes \fm{\mathbb{I}}$ under the unitary ``qubit assimilation'' mapping (Lemma \ref{lem:qam} in Appendix \ref{sec:qam}). Here $\fm{\mathbb{I}}$ acts on auxiliary $n/2$ fermionic modes. 
    
    We now show that ${\rho}$ is $(\Omega(1), \Omega(n))$-spread. Then we use this fact to lower-bound the circuit complexity of ${\rho}$.
    
    For the state $\rho$, construct $\qf{\rho}$ on $n$ qubits and $n/2$ fermions by applying inverse qubit assimilation to ${\rho}$ (replacements listed in Eqs.\,\ref{eq:qam_c_inv}). Next, trace out the $n/2$-fermion sector from $\qf{\rho}$ and thus obtain an $n$-qubit state $\qb{\rho}$. Observe
    \begin{align}
        \mathrm{Tr} (\qb{\rho} \qb{H})=\mathrm{Tr} (\qf{\rho} \qb{H}\otimes \fm{\mathbb{I}})=\mathrm{Tr} ({\rho} {H})<\epsilon n.
    \end{align}
    By definition of $\qb{H}$, this implies that $\qb{\rho}$ is $(\Omega(1), \Omega(n))$-spread. By Definition~\ref{def:well-spread}, there exists a product of single-qubit rotations $R$, a POVM $\{\pi_R(s)\}_{s\in \{0,1\}^n}$, and two sets of bitstrings $S_{1,2}\subset\bits^n$ such that $p_R(s)=\Tr (\qb{\rho} \pi_R(s))$ obeys
    \begin{align}
    \sum_{s\in S_{1,2}} p_R(s)=\Omega(1), \min_{s_{1,2}\in S_{1,2}}|s_1-s_2|=\Omega(n).
    \label{eq:qubit_side_well_spread}
    \end{align}
    Considering a general $R$ rotating $Z_j$ to $X_j \sin \theta_j \cos\phi_j +\sin\theta_j\sin\phi_j Y_j+\cos\theta_j Z_j$ for each qubit, we write it as
    \begin{equation}
    R=\prod^n_{j=1} \exp \left(\frac{i\theta_j}{2} (\sin \phi _j X_j - \cos \phi_j Y_j)\right).
    \end{equation} 
    By qubit assimilation (Lemma \ref{lem:qam}), $X_j,~Y_j,$ and $Z_j$ are mapped respectively onto $ic_{j,z}c_{j,y}$, $~ic_{j,x}c_{j,z},$ and $ic_{j,x}c_{j,y}$ for $j\in [n]$. Therefore $R$ is mapped onto
    \begin{align}
        G = \prod_j \exp \left(\frac{\theta_j}{2} \left(-c_{j,z}c_{j,y}\sin \phi_j  + c_{j,x}c_{j,z}\cos \phi_j \right) \right).
    \end{align}
    Relabeling $c_{j,y}\rightarrow c_{2j-1}$, $c_{j,x}\rightarrow c_{2j}$ and $c_{j,z}\rightarrow c_{2n+j}$, define a POVM
    \begin{align}
        \pi_G(\tilde{s})= G^\dag \left(\prod^{3n/2}_{j=1}\frac{\mathbb{I}+(-1)^{\tilde{s}_j}ic_{2j-1}c_{2j}}{2}\right) G,
    \end{align}
    where $\tilde{s}\in \bits^{3n/2}$. We adopted the notation $\pi_G(\tilde{s})$ from Definition \ref{def:gaus_measurement}, since the unitary $G$ is Gaussian.     
    We see that qubit assimilation maps $\pi_R(s)$ onto: 
    \begin{align}
        G^\dag \left(\prod^{n}_{j=1}\frac{\mathbb{I}+(-1)^{s_j}ic_{2j-1}c_{2j}}{2}\right) G
        &= \sum_{\substack{\tilde{s}_j=\pm 1,\\ j\in(n+1,..3n/2)}} G^\dag \left(\prod^{3n/2}_{j=1}\frac{\mathbb{I}+(-1)^{\tilde{s}_j}ic_{2j-1}c_{2j}}{2}\right) G\big|_{\tilde{s}_j=s_j,\,j\in[n]}
\notag\\&=\sum_{\substack{\tilde{s}_j=\pm 1,\\ j\in(n+1,..3n/2)}}\pi_G(\tilde{s})\big|_{\tilde{s}_j=s_j,\,j\in[n]}.\label{eq:qubit_to_fermion_POVM}
    \end{align}
    Note that in Eq.\,\eqref{eq:qubit_to_fermion_POVM} bits $s_j$ for $j\in(1,\ldots,n)$ came from the pre-image $\pi_R(s)$, while $s_j$ for $j\in(n+1,\ldots, 3n/2)$ are the newly introduced dummy variables. Applying the constructed POVM $\pi_G(s)$, we find that the $3n/2$-fermion state $\rho$  is $(\Omega(1), \Omega(n))$ by Definition \ref{def:well-spread}. Indeed, using $S_{1,2}$ from Eq.\,\eqref{eq:qubit_side_well_spread} we can directly construct two sets of bitstrings $\tilde{S}_{1,2}\equiv \{(s_1,..s_n,s'_1,..s'_{n/2})~| ~s\in S_{1,2},~s'\in\bits^{n/2} \}\subset \bits^{3n/2}$, such that $p_G(s)=\Tr (\rho \pi_G(s))$ obeys

    \begin{align}
    \label{eq:fermion_side_well_spread}
    \sum_{\tilde{s}\in \tilde{S}_{1,2}} p_G(\tilde{s})=\sum_{s\in S_{1,2}} p_R(s) =\Omega(1), \min_{\tilde{s}_{1,2}\in \tilde{S}_{1,2}}|\tilde{s}_1-\tilde{s}_2|=\min_{s_{1,2}\in S_{1,2}}|s_1-s_2|=\Omega(n).
    \end{align}

    Consider constructing ${\rho}$ with a depth-$T$ fermion circuit with free access to Gaussian operations, using $m=\mathrm{O}(n)$ ancillary modes. To prove that $T=\Omega(\log n)$ depth is required, we employ $({\Omega}(1), \Omega(n))$-spreadness of $\rho$ and Lemma \ref{lem:fdepth} stated below. This directly yields $T=\Omega(\log n)$. \\

    Finally, consider producing ${\rho}$ in a depth $T$ fermionic circuit $U$ (see Eqs.\,(\ref{eq:fermi_state_prep},\,\ref{eq:fermion_circuit})), starting from an arbitrary Gaussian initialization $\sigma_G$. Since the output state ${\rho}$ is obtained as a marginal of $U\sigma_G U^{\dag}$ on Majorana operators $[3n]$, we are not concerned with the action of $U$ on the entire system of Majorana operators $[3n+2m]$. Without changing ${\rho}$, $U$ can be replaced with $U'$: its backward light cone stemming from $[3n]$.
    This backward light cone will be supported on at most $3n\cdot 4^T$ Majoranas, because each gate in $W_t$ (Eq.~\eqref{eq:nongaussian_operation}) involves at most three Majoranas in addition to each one from the light cone at layer $t+1$. Naively, an arbitrary number $2m$ of ancillary Majoranas can still be non-trivially involved in the initialization state $\sigma_G$. However, one can replace this state with its marginal ${\rho}_G$ on $3n\cdot 4^T$ Majoranas supporting $U'$, without changing ${\rho}$. In turn, being a mixed Gaussian state on $3n\cdot 4^T$ Majoranas, ${\rho}_G$ can be purified back onto $O(n\cdot 4^T)$ fermions.\\
    Therefore, any $3n/2$-fermionic state prepared using a depth-$T$ fermionic circuit, arbitrary Gaussian initialization, and arbitrary $m$ ancillary fermionic modes, can be prepared in the same setting with $3n/2+m$ reduced to $O(n\cdot 4^T)$. Employing a next Lemma \ref{lem:fdepth} and $(\Omega(1), \Omega(n))$-spreadness of $\rho$ once again, we obtain $T=\Omega(\log n)$, finalizing the proof of the Theorem.    
\end{proof}

\begin{lemma}
\label{lem:fdepth}
    Consider a $(\mu, L)$-spread $l$-fermion state $\rho$. Any fermionic circuit using $m$ ancillary fermionic modes and free access to Gaussian operations that produces $\rho$ must be of depth at least 
    \begin{equation}
   T=\frac{1}{2} \log_3 \left(\frac{L^2}{1600 (l+m)\log(1/\mu)}\right).
    \end{equation}
    
\end{lemma}

\begin{proof}
    The following parallels the proof of Fact 4 from \cite{anshu2022nlts}. The key additional observation is that Gaussian circuits, however deep, do not change the locality of fermionic operators. Consider preparing $\rho$ as 
    \begin{align}        \rho=\Tr_{\left[2(m+l)\right]\bs\left[2l\right]}\left(U\sigma_G U^\dag\right),
    \end{align}
    with a fermionic circuit $U$ and the pure Gaussian starting state
    \begin{align}
    \sigma_G=\frac{1}{2^{l+m}}\prod^{l+m}_{k=1} (\mathbb{I}+i c_{2j-1} c_{2j}).
    \end{align}
    
    Consider the operator $Q=\frac{1}{(l+m)}\sum^{l+m}_{j=1} \frac{1}{2}(\mathbb{I}-i c_{2j-1} c_{2j})$. The spectrum of $Q$ is $(0,1/(l+m),2/(l+m),..,1)$, with the non-degenerate ground state $\rho_G$.
    The Hamiltonian $U Q U^\dag$ has the same spectrum as $Q$ and has $U\sigma_GU^{\dagger}$ as the ground state; we now prove that this Hamiltonian is also $2 3^T$-local, where $T$ is the depth of $U$. Indeed, (i) any Gaussian circuit $G_t$ in $U$ (cf. Eq.\,\eqref{eq:fermion_circuit_free_gauss}) does not change the locality of a fermionic operator $C_{K'}$ in $U Q U^\dag$, while (ii) every layer $W_t$ of a non-Gaussian circuit can increase the weight of $C_{K'}$ at most by a factor $3$. To show point (i), consider the two options for a transformation of any Majorana monomial $C_{K'}$ with an elementary Gaussian operation:
\begin{align}
\label{eq:Gaussian_weight_transform}
    e^{-i w_{k_1,k_2}\,c_{k_1}c_{k_2}}C_{K'}e^{i w_{k_1,k_2}\,c_{k_1}c_{k_2}}=
\begin{cases}    
    C_{K'}(\cos(2w_{k_1,k_2})+ i\sin(2w_{k_1,k_2})\,c_{k_1}c_{k_2})~&\mathrm{if}~\{C_{K'},c_{k_1}c_{k_2}\}=0,\\
    C_{K'}~&\mathrm{if}~[C_{K'},c_{k_1}c_{k_2}]=0.
    \end{cases}
\end{align}
In the first option, $k_1$ or $k_2$ must belong to $K'$ and therefore $C_{K'}c_{k_1}c_{k_2}$ has the same weight as $C_{K'}$. The second option trivially conserves locality.  To show point (ii), consider the transformation similar to Eq.~\eqref{eq:Gaussian_weight_transform}, using an elementary gate $\exp\left(i w_{K}\,C_K\right)$ with $|K|=4$; this transformation can only add weight $2$ to $C_{K'}$, and only if $\{C_{K'},C_K\}=0$. Since $W_t$ only contains non-overlapping generators $C_K$, there will be at most $|K'|$ generators such that $\{C_{K'},C_K\}=0$. Therefore, a transformation of $C_{K'}$ with $W_t$ will produce operators with a weight of at most $|K'|+|K'|\cdot 2= 3|K'|$.\\

We now use a matrix polynomial approximation $P(U Q U^\dag)$ to the ground state of degree $f$ (concretely, the polynomial construction in \cite{KahnLS96,BuhrmanCWZ99,anshu22area}) such that
    \begin{align}
        \| \ket{\psi}\!\bra{\psi} - P(U Q U^\dag) \|_\infty\leq \exp \left(-\frac{f^2}{100 (l+m)}\right).
        \label{eq:approx_projector}
    \end{align}
    As $U Q U^\dag$ is $2\cdot3^T$-local, $P(U Q U^\dag)$ is $2f 3^T$-local. Fixing $f=\frac{L}{4\cdot 3^T}$, we obtain $P(U Q U^\dag)$ that is $L/2$-local. The right-hand-side of Eq.~\eqref{eq:approx_projector} then becomes $\exp (-\frac{L^2}{1600 (l+m)\cdot 3^{2T}})$.

    Next, from $\rho$ being $(\Omega(1), \Omega(n))$-spread there exists a Gaussian $G$ and $S_1, S_2\subset \bits^l$ with $\min_{s_{1,2}\in S_{1,2}}|s_1-s_2|\geq L$ such that
    \begin{align}
        \| \sum_{s_{1,2}\in S_{1,2}} \pi_G(s_1) \ket{\psi}\!\bra{\psi} \pi_G(s_2)\|_\infty \geq \mu.
    \end{align}
    On the other hand, from the $L/2$-locality of $P(U Q U^\dag)$ it follows that
    \begin{align}
        \| \sum_{s_{1,2}\in S_{1,2}} \pi_G(s_1) P(U Q U^\dag) \pi_G(s_2)\|_\infty=0,
        \label{eq:locality_vanishing}
    \end{align}
    since $\pi_G(s_1) ~C ~\pi_G(s_2)=0$ for any $|s_1-s_2|\leq L$ and $L/2$-local $C$. The latter condition is true since the Gaussian transformation $G$ conserves the locality of operator $C$, and because this condition manifestly holds for the case $G=\mathbb{I}$. Collecting Eqs.~\eqref{eq:approx_projector}-\eqref{eq:locality_vanishing}, we arrive at:
    \begin{align}
        \exp \left(-\frac{L^2}{1600 (l+m)\cdot 3^{2T}}\right)\geq \mu,
    \end{align}
    which amounts to the claimed lower bound on $T$.
\end{proof}

\section{Qubit PCP versus Fermionic PCP}
\label{sec:PCP}

The quantum PCP complexity class (QPCP) was defined in e.g. \cite{AALV:detect}. The quantum PCP conjecture \cite{AAV} asks whether ${\rm QPCP}[q=0(1)]$ is equal to QMA, where $q=O(1)$ is the number of qubits of the proof which are checked. Here we introduce a new class ${\rm FermPCP}[q=O(1)]$ which accesses $q$ fermionic modes of a (fermionic) proof and argue that this class could be larger than ${\rm QPCP}[q=O(1)]$ due to the limitations of fermion-to-qubit mappings. \\

In the definitions below we do not make explicit reference to the number of random bits of the verifier: it is sufficient if this is at most ${\rm poly}(\log n)$ when the number of qubits (or fermionic modes) that the verifier accesses is $q=O(1)$ or $q=O(\log n)$ which is what we will use in this section. In addition, we allow the verifier to simulate randomness in the quantum circuit by having (${\rm poly}(n)$) ancillary qubits. We start by reproducing the definition of QPCP:

\begin{definition}[{\rm QPCP}]
    A promise problem $L=L_{\rm yes} \cup L_{\rm no} \in {\rm QPCP[q]}$ if there exists a quantum polynomial verifier and a polynomial $p(.)$ with the following properties. The verifier receives as input a classical string $x$ and a $p(x)$-qubit density matrix $\rho$.
The verifier randomly picks $q$ qubits out of $p(x)$ according to some scheme using her random bits, and then runs a polynomial-sized quantum circuit $V(x,\rho)$, which can use ${\rm poly}(n)$ ancillary qubits, and the circuit $V$ accesses only the chosen $q$ qubits as inputs. The accept/reject output of $V$ is obtained by a measurement in the $Z$-basis of one of the ancilla qubits. Then $L\in {\rm QPCP[q]}$ when,
\begin{itemize}
    \item If $x\in L$, $\exists \rho$ such that ${\rm Prob}(V(x,\rho)\; {\rm accepts})\geq 2/3$.
    \item If $x\not \in L$, $\forall \rho$, ${\rm Prob}(V(x,\rho)\; {\rm accepts})\leq 1/3$.
\end{itemize}
\end{definition}

Here is the definition of the analogous fermionic class:

\begin{definition}[{\rm FermPCP}]
    A promise problem $L=L_{\rm yes} \cup L_{\rm no} \in {\rm FermPCP[q]}$ if there exists a quantum polynomial fermionic verifier and a polynomial $p(.)$ with the following properties. The verifier receives as input a classical string $x$ and a $p(x)$-fermionic state $\rho_f$. The verifier randomly picks $q$ fermionic modes (i.e., $2q$ operators) out of $p(x)$ according to some scheme using her random bits, and then runs a polynomial-sized quantum circuit $V(x,\rho_f)$, possibly using ${\rm poly}(n)$ ancillary qubits, and the circuit $V$ accesses only the chosen $q$ fermionic modes as inputs. The accept/reject output of $V$ is obtained by a measurement in the $Z$-basis of one of the ancilla qubits.
Then $L \in {\rm QPCP[q]}$, when
\begin{itemize}
    \item If $x\in L_{\rm yes}$, $\exists \rho_f$ such that ${\rm Prob}(V(x,\rho_f)\; {\rm accepts})\geq 2/3$.
    \item If $x\in L_{\rm no}$, $\forall \rho_f$, ${\rm Prob}(V(x,\rho_f)\; {\rm accepts})\leq 1/3$.
\end{itemize}
\end{definition}

We note that the fermionic verifier and its access to a fermionic proof could be mapped to a qubit verifier and its proof by a standard fermion-to-qubit mapping such as the Jordan-Wigner transformation, but such a mapping clearly does not preserve the limited access structure of QPCP. Before we prove two statements about the relations between these classes, let us review what problem is complete for these classes and define the local Hamiltonian density problem:

\begin{definition}[(Fermionic) Local Hamiltonian Density Problem (LHD)]
    Let $H=\frac{1}{m}\sum_{i=1}^m H_i$ be a local $n$-qubit (resp. $n$-fermion) Hamiltonian with $H_i\geq 0$, $||H_i||=O(1)$. Either $\lambda(H)\leq a $ or $\lambda(H)\geq b $, for constants $a > 0,b > 0$ and a constant $\epsilon=b-a> 0$. The qubit (resp. fermionic) local Hamiltonian density (${\rm LHD}$) problem is to decide which is the case.
    \end{definition}

It is known that ${\rm LHD} \in {\rm QPCP}[q=O(1)]$ by the proof of \cite{KSV} and identical arguments can be made to show that fermionic ${\rm LHD} \in {\rm FermPCP}[q=O(1)]$. The idea is that the expectation value of any (fermionic) Hamiltonian term $H_i$ can be estimated using qubit ancillas. Here, one considers a spectral decomposition of $H_i=\sum_j w_{j,i} \sigma_{j,i}$ with $1 \geq \omega_{j,i}\geq 0$, and one defines a unitary which uses an additional ancilla qubit in the state $\ket{0}$ and maps 
\begin{equation}
\sigma_{j,i}\otimes \ketbra{0}\rightarrow \sigma_{j,i}\otimes \left(\sqrt{w_{j,i}}\ket{0}+\sqrt{1-w_{j,i}}\ket{1}\right)\left(\sqrt{w_{j,i}}\bra{0}+\sqrt{1-w_{j,i}}\bra{1}\right),
\end{equation}
(see \cite{KSV}). The desired expectation can be estimated by estimating the probability of obtaining the outcome `$1$' when the qubit is measured. \\

The opposite, namely that qubit LHD is hard (and thus complete) for ${\rm QPCP}[O(1)]$, has been informally proven in \cite{AAV}. For completeness, we prove it here and we extend it to a fermionic version:

\begin{proposition}
    The qubit ${\rm LHD}$ problem is ${\rm QPCP}[O(1)]$-hard by a polynomial-time quantum reduction. Similarly, the fermionic ${\rm LHD}$ problem is ${\rm FermPCP}[O(1)]$-hard by a polynomial-time quantum reduction.
    \label{prop:hard}
\end{proposition}

\begin{proof}
    We show that any ${\rm QPCP}[O(1)]$ proof system can be mapped onto a LHD problem. The verifier uses $O(\log n)$ bits to draw from some probability distribution $p_i$ and each $i$ corresponds to picking a certain subset of $q$ qubits and the number of such sets is $m={\rm poly}(n)$. Let the acceptance qubit be labelled ancilla qubit number 1. The probability for acceptance $p_{\rm accept}=\sum_i p_i p_{\rm accept}^i$ with $p_{\rm accept}^i=\frac{1}{2}(1+{\rm Tr}(O_i \rho))$ where $O_i$ is a traceless $q$-qubit observable acting only non-trivially on the chosen $q$ qubits of $\rho$. Then we construct $H$ as $H=\frac{1}{m}\sum_i H_i$ with $H_i=\frac{1}{2} p_i (\mathbb{I}-O_i)$, obeying $||H_i||\leq 1$ and $H_i \geq 0$. Hence, if $x\in L_{\rm yes}$, there exists a $\rho$ such that  $\lambda(H)\leq \frac{1}{3}$. When $x \in L_{\rm no}$, we have $\forall \rho$, $\lambda(H) \geq \frac{2}{3}$. To construct $O_i$, one needs to run the verifier's quantum circuit where we replace the $q=O(1)$ input qubits by all possible eigenstates of the Pauli operators on $q$ qubits and apply process tomography on the superoperator $\mathcal{S}_i(\rho)={\rm Tr}_{\mbox{\tiny anc. but 1}}V_i  \rho \otimes \ket{00\ldots 0}\bra{00\ldots 0} V_i^{\dagger}$, where $\ket{00\ldots 0}$ is the initial state of all ancillary qubits and $V_i$ is the circuit $V$ using the Pauli qubit operators in the chosen subset $i$. Using process tomography one constructs the $q$ to $q+1$-qubit TPCP map ${\cal S}_i()$ \footnote{Note that it is $q$ to $q+1$ since one ancilla qubit is fixed as $\ket{0}$ on input but is needed/used as acceptance output qubit.}, and thus $S_i^{\dagger}()$ (with negligible error $1/{\rm poly}(n)$), and applies it to $Z_1$ to construct $O_i$, i.e. $O_i=S_i^{\dagger}(Z_1)$.\\
    
    In the fermionic case, all goes through similarly. Let ancilla qubit $1$ be the acceptance qubit. The verifier measures $Z_1=\pm 1$ or $p_{\rm accept}^i=\frac{1}{2}(1+{\rm Tr}( Z_1 V_i \rho \otimes \ket{00\ldots }\bra{00\ldots 0} V_i^{\dagger}))=\frac{1}{2}(1+{\rm Tr}(Z_1 {\cal S}_i(\rho))=\frac{1}{2}(1+{\rm Tr}(O_i \rho))$. Here $V_i$ only uses those fermionic modes in the chosen subset $i$ and the state $\ket{00\ldots 0}$ of all ancillary qubits. To construct ${\mathcal S}_i$ which maps $q$ fermion modes to $q$ fermions and 1 ancilla qubit, we can construct its corresponding Choi-Jamiolkowski state \cite{bravyi2004lagrangian}. Let $c_{i[1]},\ldots c_{i[2q]}$ be the chosen set $i$ of Majorana operators and take an additional set of reference operators $c_{r,1}, \ldots c_{r,2q}$, then $\rho_{\rm Choi}=\mathcal{S}_i(\frac{1}{2^{2q}}\Pi_{j=1}^{2q}(\mathbb{I}+i c_{i[j]} c_{r,j}))$. From $\rho_{\rm Choi}$, a set of $O(1)$ Kraus operators $A_{k,i}$, --which can be expressed as even polynomials in the Majorana operators $c_{i[1]},\ldots c_{i[2q]}$ and single ancilla qubit Paulis--, can be obtained \cite{bravyi2004lagrangian}, and hence $O_i$ can be constructed and thus $H$ can be constructed via a quantum reduction.    
\end{proof}

We note that the qubit LHD problem is not clearly hard for the class ${\rm FermPCP}[O(1)]$. In the following Lemmas we establish some relations between ${\rm QPCP}$ and ${\rm FermPCP}$:

\begin{lemma}
${\rm QPCP}[O(1)] \subseteq {\rm FermPCP}[O(1)]$.     
\end{lemma}

\begin{proof}
Let $L \in {\rm QPCP}$, and so we seek to construct a fermionic proof system for $L$. First, if $p(x)$ is odd, add an extra qubit in $\ket{0}$ to the proof so that $p(x)$ is even. In addition, wlog we assume that $q$, the number of access qubits is even. We use the unitary qubit-assimilation mapping $\mathcal{M}$ of Lemma \ref{lem:qam} in Appendix \ref{sec:qam} to map the $k=p(x)$-qubit state $\rho$ onto a state $\rho_f$ on $3k$ Majorana operators ($3k/2$ fermionic modes), that is, if we expand the witness state $\rho$ in terms of Pauli $X_j,Z_j$ and we replace those by Eqs.~\eqref{eq:qam_x},~\eqref{eq:qam_z}. Similarly, each occurrence of $X_j$ or $Z_j$ where $j$ is one of the chosen input qubits in the circuit of the verifier is replaced using Eqs.~\eqref{eq:qam_x},~\eqref{eq:qam_z}. The ancillary qubits remain as is. From the qubit proof, the verifier selects $q$ qubits, and this thus corresponds to selecting $3q$ Majorana operators ($3q/2=O(1)$ fermionic modes). Thus we can map the qubit proof for $x\in L_{\rm yes}$ directly onto a fermionic proof and the acceptance probability is the same. Note that the additional fermionic Hilbert space in Lemma \ref{lem:qam} is simply not used in this conversion.\\

When $x\in L_{\rm no}$, the verifier proceeds identically using the mapping $\mathcal{M}$, but the prover may provide an arbitrary fermionic state $\rho_f$ of $3k/2$ fermionic modes as input, not obeying the mapping $\mathcal{M}$. So we need to argue that this state can always be mapped to some qubit state $\rho_q$, such that if the fermionic verifier is fooled in accepting by this state, then so will the qubit verifier with input $\rho_q$, which was excluded by definition. The state $\rho_f$ will be an even polynomial in the operators $c_{y,j},c_{z,j}$ and $c_{x,j}$. Using the $\tilde{c}_j$ operators in Eq.~\eqref{eq:qam_c}, these operators can be mapped to qubit Pauli operators $X_j, Z_j$ and $\tilde{c}_j$, hence obtaining a state $\rho_{qf}$. Then we take $\tilde{\rho}_q={\rm Tr}_f(\rho_{qf})$, the partial trace over the fermionic system to get a $n$-qubit state $\rho_q$ (effectively this means omitting any term which involves the operators $\tilde{c}_j$).
Note that since the fermionic verifier uses ${\cal M}$ the fermionic gates and final measurement of the fermionic verifier never use the operators $\tilde{c}_j$ and hence tracing out has no effect on the action of the verifier and the selection of $3q/2$ fermionic modes corresponds precisely to the selection qubits of $q$ qubits in $\rho_q$. Thus if $\rho_f$ would lead the fermionic verifier $V$ to accept with probability $> 1/3$, then the qubit verifier would accept with probability $> 1/3$ on the state $\rho_q$, in contradiction with the qubit proof system.
\end{proof}

\begin{lemma}\label{lem:ferm_to_q}
${\rm FermPCP}[O(1)] \subseteq {\rm QPCP}[O(\log n)]$.     
\end{lemma}

\begin{proof}
Let $L \in {\rm FermPCP}$, and so we seek to construct a qubit proof system for $L$. Using the Bravyi-Kitaev transformation \cite{bravyi00ferm}, we map the $k=p(n)$-fermionic mode state $\rho_f$ onto a $k$-qubit state and the chosen $2q$ Majorana operators $c_j$ are replaced by products of $O(\log (p(n)))$ single-qubit Pauli operators, hence requiring access to $O(\log n)$ qubits. Thus we can map the fermionic qubit proof for $x\in L_{\rm yes}$ directly onto a qubit proof, albeit with logarithmic access. When $x\in L_{\rm no}$, the prover may provide an arbitrary $k$-qubit state $\rho_q$ as input, not necessarily obeying the Bravyi-Kitaev transformation, while the verifier applies the circuit obtained through the mapping. Applying the inverse of the Bravyi-Kitaev transformation to the state $\rho_q$ leads to fermionic state $\rho_f$ of $k$ fermionic modes. If the qubit verifier accepted on $\rho_q$ with probability $> 1/3$, then applying the inverse would have lead to the fermionic verifier accepting with probability $> 1/3$, which was excluded. 
\end{proof}

Note that the qubit (resp. fermionic) LHD problem which is complete for ${\rm QPCP}[O(1)]$ (respectively ${\rm FermPCP}[(O(1)]$) is not necessarily sparse. Previous results (Theorem 13 in \cite{BH:dense-approx}) give a polynomial-time classical algorithm for the {\em dense} qubit case, so that problem certainly cannot be $\mathrm{QMA}$-hard. 
It is an open question whether the \emph{dense} fermionic LHD problem can be $\mathrm{QMA}$-hard. Finally, whether ${\rm FermPCP}[O(1)] ={\rm QPCP}[(O(1)]$ is an open question, which relates to (1) whether the {\em sparse} fermionic LHD problem is ${\rm FermPCP}[O(1)]$-hard (note that Proposition \ref{prop:hard} merely shows that the general fermionic LHD problem is hard), and (2) the existence of fermion-to-qubit mappings for sparse fermionic interactions which do not introduce non-local constraints. The latter topic is discussed in Section \ref{sec:discuss}.

\section{Acknowledgements}
We thank Joel Klasssen and Sergei Bravyi for discussions on fermion-to-qubit mappings. Y.H. and B.M.T. acknowledge support by QuTech NWO funding 2020-2024 – Part I “Fundamental Research”, project number 601.QT.001-1, financed by the Dutch Research Council (NWO). Y.H and J.H. acknowledge support from the Quantum Software Consortium (NWO Zwaartekracht).  AA acknowledges support through the NSF CAREER Award No. 2238836 and NSF award QCIS-FF: Quantum Computing \& Information Science Faculty Fellow at Harvard University (NSF 2013303).

\section{Discussion}
\label{sec:discuss}

In this paper we constructed fermionic NLTS Hamiltonians from qubit Hamiltonians with well-spread low-energy states. An interesting open question is to consider the reverse construction. It is known that one cannot map all $O(1)$-local Majorana operators to qubit operators of weight less than $\Omega(\log(n))$ (see e.g. the argument in \cite{HSHT}), i.e. the Bravyi-Kitaev construction used in \cref{lem:ferm_to_q} has the optimal scaling in $n$. It must be noted though, that for \emph{sparse} fermionic Hamiltonians there exists the Bravyi-Kitaev superfast encoding \cite{bravyi00ferm}, which unitarily maps local fermionic Hamiltonians terms onto $O(1)$-local qubit interactions. However this mapping is only valid in a subspace specified by a set of non-local stabilizer generators. When mapping a fermionic NLTS Hamiltonian to a qubit NLTS Hamiltonian, these generators would have to be included in the qubit Hamiltonian. In general these generators are not $O(1)$-local, nor would the resulting Hamiltonian necessarily be sparse. Adapting the BK superfast encoding to avoid these two properties is the topic of active research in quantum simulation (see \cite{chien2023simulating} and references therein). In \cref{sec:sf_encoding} we describe a construction that achieves both of these objectives, i.e. it encodes a sparse local fermionic Hamiltonian into a sparse local qubit Hamiltonian, through an adaptation of the BK superfast encoding. However it requires $\Theta(n^2)$ qubits to encode $n$ fermionic modes, making it useless when constructing NLTS Hamiltonians. This leaves open the problem of finding a way to encode an $n$-mode sparse $O(1)$-local fermionic Hamiltonian into an $O(n)$-qubit sparse $O(1)$-local qubit Hamiltonian. 
Given our repeated failures at finding such a construction we believe there might be some fundamental obstruction. If ${\rm FermPCP}[O(1)]$ is strictly larger than ${\rm QPCP}[O(1)]$, pursuing a fermionic PCP theorem might be a fruitful endeavor. Finally, we note that PCP versions of QCMA \cite{WFC:guide} are equivalent when formulated in terms of qubits or fermions.\\

This paper does not make progress towards understanding the quantum PCP conjecture, but shows that one can construct fermionic Hamiltonians with no trivial low energy states. Extending such a result beyond Majorana fermion stabilizer Hamiltonians would be of some interest. On the more practical side, a natural next step would be to give more realistic families of fermionic Hamiltonians without low-energy trivial states. Developing numerical tools in order to diagnose the fermionic NLTS property in realistic Hamiltonians is another direction of practical interest.

\appendix

\section{Qubit assimilation mapping}
\label{sec:qam}

We detail the mapping used in Theorem\,~\ref{thm:fNLTS}. See Ref.\,\cite{BGKT:manybody} for an earlier use of this mapping.
\begin{lemma}
\label{lem:qam}
    For even $n$, consider the Hilbert space $\mathcal{H}^{\rm{(q-f)}}$ of $n$ qubits (with Paulis $X_j$, $Y_j$ and $Z_j$ for $j\in[n]$) and $n/2$ fermions (operators $\tilde{c}_j$ for $j\in[n]$)), and the Hilbert space $\mathcal{H}^{\rm{(f)}}$ of $3n/2$ fermions ($c_{x,j},~c_{y,j},~c_{z,j}$ for $j\in[n]$). There is a unitary map from $\qf{\mathcal{H}}$ to $\fm{\mathcal{H}}$, defined by its action on the generating operators
    \begin{align}
        X_j& \mapsto i c_{y,j} c_{z,j},\label{eq:qam_x}\\
        Z_j& \mapsto i c_{x,j} c_{y,j},\label{eq:qam_z}\\
        \tilde{c}_j& \mapsto i c_{x,j} c_{y,j} c_{z,j}. \label{eq:qam_c}
    \end{align}
\end{lemma}

\begin{proof}

Consider the bases of Hermitian operators in $\mathcal{H}^{\rm{(f)}}$ and $\qf{\mathcal{H}}$,
\begin{align}
   \label{eq:fm_pauli_basis}\mathcal{P}^{\rm{(f)}}(\mathbf{s}=\{s_{\alpha,j}, \alpha\in\{x,\,y,\,z\}\})&=i^{\frac{|\mathbf{s}|(|\mathbf{s}|-1)}{2}}\prod_j c^{s_{x,j}}_{x,j}c^{s_{y,j}}_{y,j}c^{s_{z,j}}_{z,j},\\
   \label{eq:qf_pauli_basis}
   \qf{\mathcal{P}}(\mathbf{s}^{x}=\{s^x_{j}\}, \mathbf{s}^{z}=\{s^z_{j}\}, \mathbf{s}^c=\{s^c_{j}\})&=i^{\mathbf{s}^{x}\cdot\mathbf{s}^{z}+\frac{|\mathbf{s}^c|(|\mathbf{s}^c|-1)}{2}} \prod_j X_j^{s^x_{j}} Z_j^{s^z_{j}} \tilde{c}^{s^c_{j}}_{j}.
\end{align}
Both $\mathcal{H}^{\rm{(q-f)}}$ and $\mathcal{H}^{\rm{(f)}}$ are isomorphic to the $3n/2$-qubit Hilbert space (denote it as $\qb{\mathcal{H}}$) via the Jordan-Wigner transformation. Under this isomorphism, the Pauli basis of Hermitian operators in $\qb{\mathcal{H}}$ is equivalent to the Hermitian bases in Eqs.~\eqref{eq:fm_pauli_basis}-\eqref{eq:qf_pauli_basis}. Sets $\{c_{x,j},~c_{y,j},~c_{z,j}\}$ and $\{X_{j},~Z_{j},~\tilde{c}_{j}\}$ are two alternative sets of generators of the Pauli group, i.e., independent Pauli strings in $\qb{\mathcal{H}}$. The mapping (denote it as $\mathcal{M}$) between these generators defines the mapping on the whole Pauli group,
\begin{align}
\label{eq:qam_pauli_string_map}
    \mathcal{M}\left(\prod_j X_j^{s^x_{j}} Z_j^{s^z_{j}} \tilde{c}^{s^c_{j}}_{j}\right)=\prod_j\mathcal{M}\left(X_j^{s^x_{j}}\right)\mathcal{M}\left(Z_j^{s^z_{j}}\right)\mathcal{M}\left(\tilde{c}^{s^c_{j}}_{j}\right).
\end{align}
Observe that for two elements $\mathcal{P}_A$ and $\mathcal{P}_B$ of the Pauli group defined in terms of $\{X_{j},~Z_{j},~\tilde{c}_{j}\}$ as in Eq.\,~\eqref{eq:qf_pauli_basis}, the map $\mathcal{M}$ has the property 
\begin{align}
\label{eq:qam_pauli_product}
\mathcal{M}\left(\mathcal{P}_A\right)\mathcal{M}\left(\mathcal{P}_B\right)=\mathcal{M}\left(\mathcal{P}_A \mathcal{P}_B\right).
\end{align}
This property follows from Eq.\,~\eqref{eq:qam_pauli_string_map} and the fact that $\mathcal{M}(X_{j})$, $\mathcal{M}(Z_{j})$ and $\mathcal{M}(\tilde{c}_{j})$ mutually commute or anticommute in the same way as $X_{j}$, $Z_{j}$ and $\tilde{c}_{j}$ (among themselves and across different $j$). The property in Eq.~\eqref{eq:qam_pauli_product} implies that the unitary map claimed in the Lemma exists, defined as a Clifford transformation on $\qb{\mathcal{H}}$ \cite{CRSS}.
\end{proof}

Since the map of Lemma \ref{lem:qam} is unitary, it admits an inverse from $\fm{\mathcal{H}}$ to $\qf{\mathcal{H}}$, which acts on the fermionic operators as
\begin{align}
        c_{x,j}\mapsto X_j \tilde{c}_j,~~
        c_{y,j}\mapsto -Y_j \tilde{c}_j,~~
        c_{z,j}\mapsto Z_j \tilde{c}_j.
        \label{eq:qam_c_inv}
\end{align}
A curious observation (which we do not use in this work) is that the mapping of $\fm{H}$ onto $\qf{H}$ using Eqs.\,\eqref{eq:qam_c_inv} can be repeated until all fermions are replaced with qubits.
This repetition procedure can be done in a variety of ways, each giving a full fermion-to-qubit mapping. In the worst case, the scaling features of the resulting mapping resemble those of Bravyi-Kitaev transformation \cite{bravyi00ferm}.
In practical applications, one may pick a variation of the procedure which leads to the best locality for the resulting qubit Hamiltonian.

\section{Sparse and local superfast encoding}\label{sec:sf_encoding}
The Bravyi-Kitaev superfast encoding is a method for mapping sparse and local fermionic Hamiltonians to sparse and local qubit Hamiltonians (up to the enforcement of certain stabilizer constraints). Here we give a version of this encoding which makes sure these stabilizer constraints are themselves local and sparse. Briefly, the BK encoding proceeds as follows. 
Given an $n$-fermion Hamiltonian $H$ we construct a graph $G$ with $n$ vertices (one for each fermionic mode) and place an edge if a term in the Hamiltonian involves both modes. More precisely, one can define vertex and edge operators as
\begin{align}
    V_i = c_{2i-1}c_{2i} \hspace{5em} i \in V(G) = [n],\\
    E_{i,j}=  c_{2i}c_{2j} \hspace{5em} (i,j) \in E(G).
    \label{eq:edge}
\end{align}
The Hamiltonian $H$ can then be recovered as an element of the ``graph algebra" generated by the above operators. We can represent this graph algebra on qubits by placing a qubit on each edge of the graph and defining a map $\mathcal{B}$ taking the $E$ and $V$ operators to Pauli operators which are local with respect to the graph $G$ (see e.g. \cite{setia2019superfast} for an explicit description). However the graph algebra comes with the non-trivial constraint that the product of the edge operators $E_{i,j}$ along any cycle of $G$ is equal to the identity, as is clear from Eq.~\eqref{eq:edge}. For $\mathcal{B}$ to be a proper algebra homomorphism we must thus restrict its image to the subspace where this is true also for the operators $\mathcal{B}(E_{i,j})$. It turns out these constraints all commute and thus form a stabilizer group generated by the products along a cycle basis of $G$.\\

The interesting question is whether these stabilizer generators can be made sparse and local. This corresponds to choosing a cycle basis for the graph $G$ that contains only cycles of constant length, and where each edge only participates in a constant number of basis cycles. It is clear that there are graphs of bounded degree for which no such basis exists. Consider for instance the family of $n$-vertex bounded-degree expander graphs given in \cite{margulis1982explicit}, which have girth $\Omega(\log(n))$, and hence have no cycles of constant length. Moreover the total length of any basis must then be $\Omega(n\log(n))$ which, from a pigeonhole argument, means that there is at least one edge that is present in $\Omega(\log(n))$ basis cycles.\\

Here we give a construction that solves both of these problems, by constructing from the graph $G$ a larger graph $\hat{G}$ that has a cycle basis of short cycles that use every edge only a constant number of times, and has $G$ as an induced subgraph. Note that this means that the graph algebra of $G$ is a subalgebra of the graph algebra of $\hat{G}$ (by considering a subset of the generators), and hence this provides a valid mapping of the original fermionic Hamiltonian to a qubit Hamiltonian. The downside of this construction is that the graph $\hat{G}$ is of size $\Theta(n^2)$, which makes it difficult to use it for NLTS-style arguments where the scale of the system ($n$ vs $n^2$) is important. However the construction might be of use in quantum simulation, and as a starting point for more sophisticated constructions with better parameters. 
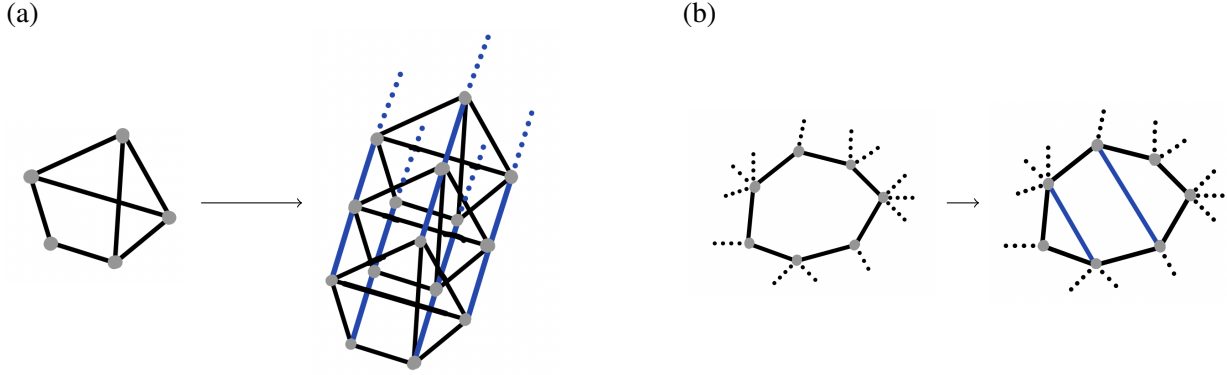
\begin{figure}
\begin{tikzpicture}
\node at (-1,2.5) {\text{(a)}};
\node at (8,2.5) {\text{(b)}};
\node (A) at (0,0) {\includegraphics[width=0.15\linewidth]{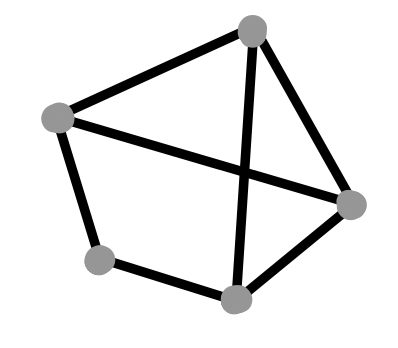}};
\node (B) at (4.5,0)  {\includegraphics[width=0.2\linewidth]{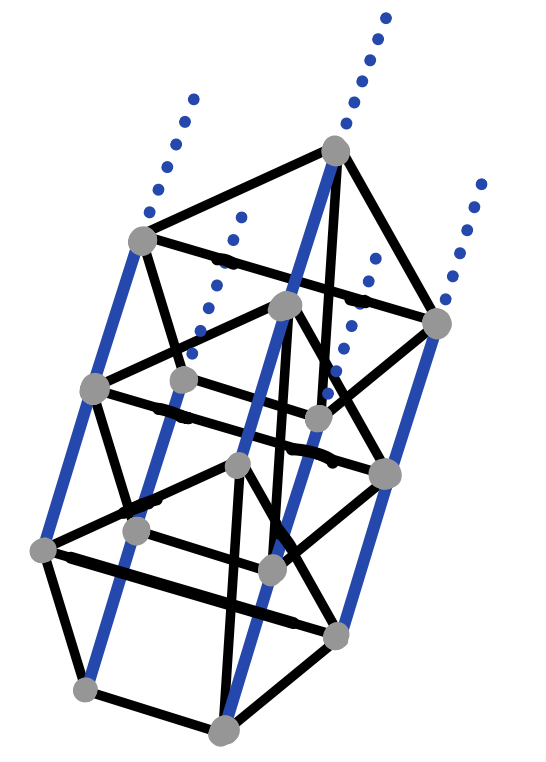}};
\draw [->] (A) to (B);
\node (C) at (9.5,0)  {\includegraphics[width=0.2\linewidth]{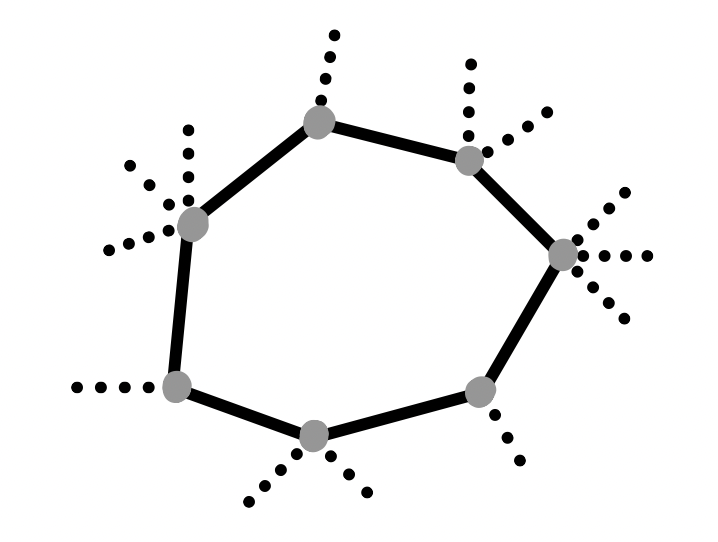}};
\node (D) at (13.5,0)  {\includegraphics[width=0.2\linewidth]{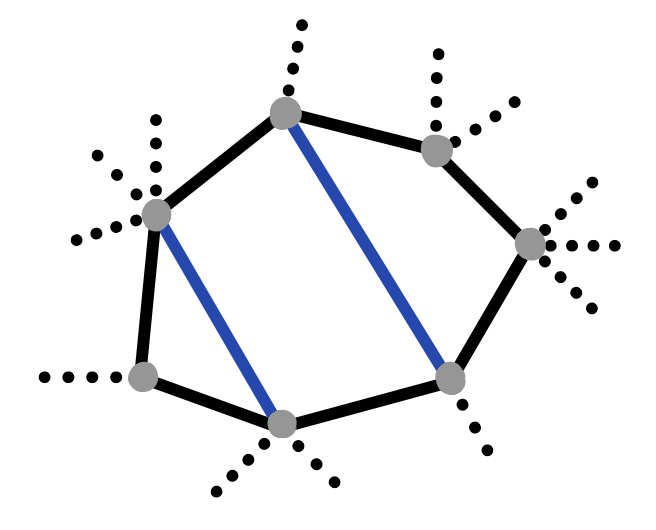}};
\draw [->] (C) to (D);
\end{tikzpicture}
\caption{\textbf{(a)}: Illustration of the stacking process described in the proof of \cref{lem:graph} with a graph on five vertices. The blue edges are the new edges added in between copies of the graph. \textbf{(b)}: Illustration of the sewing process in the proof of \cref{lem:graph}. A single large cycle is shown inside a graph, as well as its sewn up version. The blue edges are again added in the process.}\label{fig:graph_process}
\end{figure}

\begin{lemma}\label{lem:graph}
   Consider a bounded-degree connected graph $G$ on $n$ vertices. There exists a (polynomial time construable) connected graph $\hat{G}$ on $O(n^2)$ vertices with $G$ as an induced subgraph that has a cycle basis consisting of cycles of length at most $4$ which uses no edge in $\hat{G}$ more than $4$ times. 
\end{lemma}
\begin{proof}
We will explicitly construct $\hat{G}$. First, compute a minimum length cycle basis $C$ for $G$ (for instance through Horton's algorithm, which takes $O(n^4)$ time). Since $G$ has bounded degree and is connected, the cycle basis has $E(\hat{G}) -V(\hat{G})+1 = O(n)$ elements. Order the cycles in $C$ in some arbitrary way. We now construct the graph $\hat{G}$ as follows. For each element of $C$ we make a copy of the graph $G$. We take these graph copies and ``stack" them on top of each other, connecting each vertex in a graph copy to the corresponding vertex in the copies directly above and below (see \cref{fig:graph_process}(a) for an illustration). This creates $E(G) (|C|-1)$ ``vertical" cycles of length $4$. It is easy to see that the vertical cycles form an independent set, since every cycle contains an edge that is not used by any other vertical cycle. Furthermore, the set of cycles $C$ is still an independent set in $\hat{G}$. Furthermore, the union of these two sets is also independent. This union is in fact a basis for $\hat{G}$, which one can see (by direct calculation) that the dimension of the cycle space of $\hat{G}$ (i.e., $E(\hat{G}) -V(\hat{G})+1$) precisely matches the number of vertical cycles plus the dimension of the cycle space of $G$. 

Continuing our construction, consider for each cycle in the set $C$ the associated copy of $G$. In this copy, ``sew" up the cycle by adding edges across the cycle, in the manner illustrated in \cref{fig:graph_process}(b). For each cycle $c$ in $C$ this creates a number of cycles of length $3$ or $4$. Note that we add $\ceil{|c|/2}$ edges to $\hat{G}$. Since the total cycle length of a minimum length cycle basis is $O(n\log(n))$ \cite[Theorem $4.4$]{kavitha2009cycle}, we end up adding at most $O(n\log(n))$ edges. 
This completes the construction of $\hat{G}$. 
Note also that the degree of $\hat{G}$ is at most three higher than the degree of $G$.

We now propose the following basis for the cycle space of the graph $\hat{G}$. We take all vertical cycles, and all the short cycles created by the sewing procedure for each cycle in $C$. And by our earlier argument the union of the vertical cycles and the cycle basis of $G$ was a basis for $\hat{G}$ before the sewing procedure. Since the sewing procedure adds an independent cycle for each edge it adds, the resulting set is a basis for $\hat{G}$.

Clearly every cycle in this set is of length no longer than $4$. Moreover, every edge is used at most a bounded number of times by each basis cycle. For the vertical edges it is clear that their number of uses is bounded by the degree of $G$ (as they only participate in ``vertical" cycles), and the edges in each copy of the graph participate only in the $3$ and $4$ cycles making up the sewed cycle for that copy, as well as at most two vertical cycles. Hence the number of uses of each edge in the basis is also bounded.
\end{proof}

\newcommand{\etalchar}[1]{$^{#1}$}

\end{document}